\documentclass[11pt,a4paper]{article}

\usepackage{graphicx,amsmath,amssymb,amsthm,mathabx,booktabs}
\usepackage{longtable}
\date{}

\author{Yong Tan\\
entermarket@163.com}
\title{Find an Optimal Path in Static System and Dynamical System within Polynomial Runtime}

\theoremstyle{plain}
\newtheorem{theorem}{Theorem}
\newtheorem{inference}{Inference}
\newtheorem{definition}{Definition}

\begin{document}
\maketitle
\begin{abstract}
We study an ancient problem that in a static or dynamical system, sought an optimal path, which the context always means within an extremal condition. In fact, through those discussions about this theme, we established a universal essential calculated model to serve for these complex systems.  Meanwhile we utilize the sample space to character the system. These contents in this paper would involve in several major areas including the geometry, probability, graph algorithms and some prior approaches, which stands the ultimately subtle linear algorithm to solve this class problem. Along with our progress, our discussion would demonstrate more general meaning and robust character, which provides clear ideas or notion to support our concrete applications, who work in a more popular complex system.
\end{abstract}
~\newline

\section{Introduction.}
It is clear for those autonomous machines that we particularly desire them can roam in a smart way. However human has deemed the problem might be classed to finding out the optimal lane in any scene. On more popular meanings, that is to plot a perfect decision-lane within least cost in a complex system. For example, we plot navigation for a commuting in a city where it is often in a vehicle jam; let many robots cruising in a narrow and twisted space without crowd in a certain channel.

It is absolutely unfair to say there has not any available algorithm attempt to solve this class problem before this paper. But while we survey the context about those methods throughout, we should have these prime problems issue: these approaches are likely to only work in a \emph{static system} with crafted weights; or there need crafted a group natural numbers to tick endpoints; even to state to refuse a complex entry but cannot define what is a complex instance. Such weird requests cannot be clarified to us. In a word, our followed job must face to a vague context about past methods, and use the new theories and trials to clarify those amid implicit relationships, which ever puzzled our minds. Here author has not got his own mind to make a satire to these giant and memorable predecessors, moreover in this paper author are ongoing definitely and masterly to use these milestone approaches created by those pioneers, but for needing to comment their geometric meaningfulness entirely. Another our task is let our solution more \emph{general significance}.

But for the fact, we will not eventually solve all problems in this paper, actually that is impossible and we definitely need much more theories to support our works. Here we have to utilize the sample space, so that we can use these small instances to reveal some attributes of computing on complex system. Hence this is only a new beginning to solve this class problem by a new universal approach.

Well, in this paper to expand these themes, author will strictly be on the way with respected to such procedure: introducing these algorithms; raising up a conclusion and proof; examine the conclusion with trials; we might insert some contents about optimized method; finally issue discussion to finish this loop.

\section{Preliminary}
Primarily, we reserve some conventional characters, letters $n, m, V$ and $G$. Well, while we say a graph or instance, it may be represented by the form $G=(V,\tau)$\cite{1}, which we denote the collection of instance endpoints by $V$ and having form $n=\vert V\vert$ to present the cardinality about set $V$. The Greek letter $\tau$ is represented a binary relation among those endpoints, that if there is a morphic arc from $u$ to $v$, we can denote the arc by $u\tau v$ or $(u,v)\in\tau$. Thus we might make a set partition according upon this logical structure with a Cartesian product $s(u)=R(u)\times L(u)$. For this data structure we call \emph{unit subgraph}\cite{1}, which figures a structure out to a star-tree image with root to its own neighbors we call them \emph{leaves}. They are surrounding root as those targets to which, \textquotedblleft arrows\textquotedblright are shot off from root. At the encode level the mathematical structure can be dealt to a 2 dimension array in memory which the root may be an index. The letter $m$ naturally is used to denote the cardinality of a leaf collection like $m=\vert L(u)\vert$.

Otherwise, to above structure of unit subgraph, we either might describe another structure as Cartesian product $\beta(u)=L(u)\times R(u)$, which becomes to the leaf shoot the root with arrow, we call \emph{visiting subgraph}\cite{1}.

The relationship between the two structures might represent some distinct types of instance.

\begin{flushleft}
\begin{enumerate} 
\item There is a $uncertain$ state on instance, we name it \emph{mixed}, likes $\forall(x,y)\in s(x)\nLeftrightarrow(\exists(y,x)\text{ or }\nexists(y,x))\in\beta(x)$.
\item If the proposition $\forall(x,y)\in s(x)\Leftrightarrow\forall(y,x)\in\beta(x)$ holds, we call this type \emph{simple}.
\item Else if $\forall (x,y)\in s(x)\Rightarrow\forall(y,x)\notin\beta(x)$, we say it \emph{directed}.
\end{enumerate}
\end{flushleft}
~\newline
In this paper, we will not particularly or strictly appoint a certain type for instance, that it sounds as possible to satisfy the entry regarded to general or universal meaning. But we need to claim that we merely study the finite graph. Upon the above definition, the \emph{path} in figure means a queue with some distinct arcs having $\tau_{i-1}(2) =\tau_{i}(1)$ for $\tau_{i-1},\tau_{i}$ in path.

If each arc on instance is ticked by numbers to present some certain meanings, we call \emph{weighted figure}, denoted by a triple array $G=(V,\tau,W)$. The phrase \emph{total weight} is exclusively to mean a path $P$ with the sum of weight on each segment at it, which may be computed by the form.

\begin{equation}\label{1}
 W(P)=\sum_{i=1}^{N}\Vert w(\tau_{i})\Vert;\quad\text{for } P=(\tau_{i})_{i=1}^N.
 \end{equation}
~\newline
Furthermore while we refer the \emph{optimal path} from a source to a target on an instance, the context means that there is a path in a path set with \emph{minimum total weight}, else the phrase \emph{shortest path} means this item contains \emph{minimum cardinality} among those members from source to target. Of course, we can further regard the fact that an extreme path might be classed to either one amid two types or both two. Default, we regard these weights as positive volume and greater than 0.\\
~\newline
\textbf{Graph Partition. }We roughly depict the approach as making a homogenous copy for instance, as if the method dispatches the endpoints to a group ordinal components, we denoted by $R=(r_{i})_{i=1}^k:1\leq k\leq n$, and we call the component \emph{region}. We might have the following forms to character it.

\begin{flushleft}
\begin{enumerate} 
\item $R\subseteq V$ and $R=(r_{i})_{i=1}^k:1\leq k\leq n$;
\item $\forall r_{i}\in R \Longrightarrow r_{i}\neq\varnothing$ and $\forall r_{i},r_{j}\in R\Longrightarrow r_{i}\bigcap r_{j}=\varnothing$ for $i\neq j$;
\item $\forall r_{i-1},r_{i}\in R$ having $\forall v\in r_{i}\Longrightarrow \exists u\in r_{i-1}$ and $(u,v)\in\tau$.
\end{enumerate}
\end{flushleft}
~\newline
This partition structure issues a layout with several features: we might call each endpoint in the first region $source$; if there any endpoint disconnected with sources, then it would be impossible to lay in any region made by this method, thus this feature might take a condition to check a singleton node in for whether it is in connected. So we need not specially claim the connected attributer for an instance likes conventional statement. Likewise if the first region is a singleton component with an endpoint, then we can prove the conclusion that there is at least a shortest path \emph{p} between source and another endpoint $v$ as a target in $i$th region for $1<i\leq k$, there having $i-1=\vert p\vert_{\in\tau}$, at where we use the subscript term to say each member in $p$ is arc. The cardinality of path $\vert p\vert$ by default still means the amount of endpoints. The proof about the \emph{shortest path theorem} is in the paper\cite{1}. We plan to introduce this algorithm in next section again, because it is nearly to involve in the ergodic operation. The follow diagram shows the homogenous copy made by partition method on the dodecahedron.

\begin{center}
\includegraphics[height=35mm]{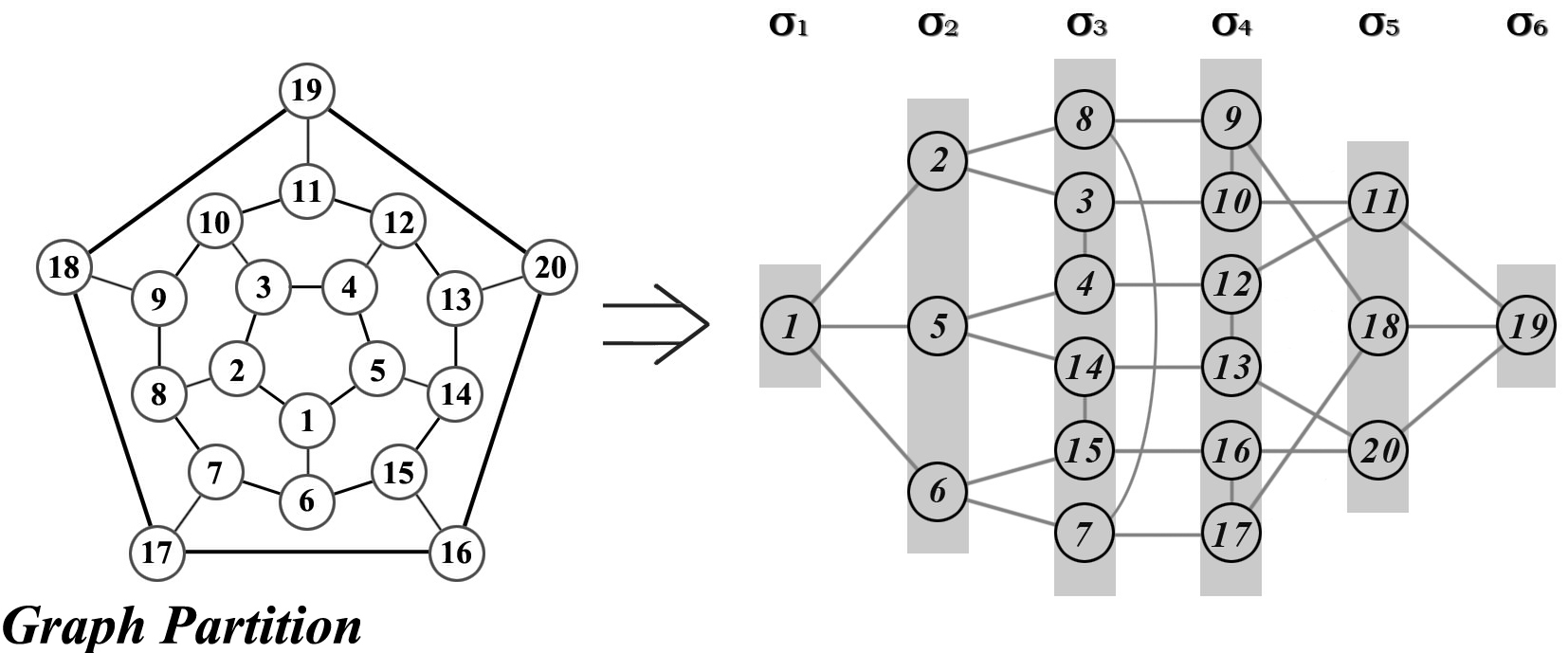}\\
Figure 1
\end{center}

\section{Algorithms}
This section we introduce these algorithms, and first one is graph partition. Here if we observe the description of graph partition, we might sound that for a unit subgraph such that the leaves would be departed to the $fore$, $current$ and $later$ region. Consequently, if any leaf is dispatched to a region amid them, then it should be in the complement of what the leaf collection $L$ differ to other two regions. Thus we can let the program iteratively divide the leaves to next region till not any leaf in leaf set for partition. We issue the algorithm in following.

\renewcommand\arraystretch{1.2}
\begin{longtable}{p{120mm}}
\caption{\textbf{Graph Partition Algorithm.}}\\
\toprule
\textbf{Input} $\tau=\{s_{i} \}_{i=1}^n$ for $s_{i}=R(u)\times L_{\in s} (u)$.\\
\textbf{Other Variables} $c=1$; $W\leftarrow s, s$ is the source.\\
\textbf{Output} $R=(r_{i})_{i=1}^{k\leq n}$.\\
\midrule
001.\quad\textbf{LOOP} $W\neq\varnothing$ \textbf{Do}\\
002.\quad\quad\textbf{For} $\forall x\in W$ \textbf{Do} $\alpha \leftarrow v:\forall v\in L_{\in s}(x)$ and $r_{\in R}(c)\leftarrow x$; \textbf{End}.\\
003.\quad\quad $\alpha \coloneq (\alpha\setminus W)\setminus r_{\in R}(c-1)\colon~ r_{\in R}(c-1)=\varnothing$ if $c\leq 0$;\\
004.\quad\quad \textbf{If} $\alpha\neq\varnothing$ \textbf{Than} $W\coloneq\alpha$ and $c\coloneq c+1$; \\
005.\quad \textbf{Go to LOOP}.\\
\bottomrule
\end{longtable}
~\newline
About the runtime complexity for this algorithm, however we can regard the cardinality of a region as $\vert r\vert=\lambda n$ for $0<\lambda\leq 1$, then the amount of regions may be $1/\lambda$, and the amount of leaves dabbling in a region might be many of $\lambda nm$. The program needs to remove those members off it, which have been in the fore and current regions, whose size both might be $\lambda n$. Therefore the complexity term might be $\lambda^2 n^2 m$, but for need to remove the repeated members in practice, the term might be $(\lambda mn)^2$ eventually. Hence the algorithm complexity is $(\lambda mn)^2\cdot (1/\lambda)=O(E^2)$ for $E=mn$. 

~\newline
\textbf{Optimization. }If we further observe above operation carefully, then we can identify the fact: each leaf set would be screened for only once. If we have two hashed table to record the states for each leaf, including dispatch and remove, then the program might save the time for the operations of comparison jobs among those arrays. Of course we have to reform the label on each endpoint to a group of continual natural number for compute the offset in array when program has the operation on those members. Here we individually setup two hash tables $T$ and $T^{\prime}$, and the partition algorithm should be updated in following.\\

\renewcommand\arraystretch{1.2}
\begin{longtable}{p{120mm}}
\caption{\textbf{Optimize Graph Partition Algorithm.}}\\
\toprule
\textbf{Input} $\tau\colon \vert \tau\vert =n$, $\tau[i]=L_{\in s}(i)$ for $i> 0$ and $i\in V$.\\
\textbf{Hash Tables } $T,T^{\prime}\colon T[i]=T^{\prime}[i]=0$ and $\vert T\vert=\vert T^{\prime}\vert=n$;\\
\textbf{Output} $R=(y_{i})_{i=1}^{k\leq n}\colon y_i\in V$.\\
\midrule
001.\quad\textbf{LOOP} $W\neq\varnothing$ \textbf{Do}\\
002.\quad\quad\textbf{For} $\forall x\in W$ \textbf{Do} $T[x]=c$ and $R\leftarrow x$; \textbf{End}.\\
003.\quad\quad\textbf{For} $\forall x\in W$ \textbf{Do}\\
004.\quad\quad\quad\textbf{For} $\forall u\in L_{\in s}(x)$ \textbf{Do}\\
005.\quad\quad\quad\quad \textbf{If} $T^{\prime}[u]=0$ and $T[u]=0$ \textbf{ Than } $r\leftarrow u$ and $T^{\prime}[u]=1$;\\
006.\quad\quad\quad\textbf{End}.\\
007.\quad\quad\textbf{End}.\\
008.\quad\quad \textbf{If} $r\neq\varnothing$ \textbf{Than} $W\coloneq r$ and $c\coloneq c+1$;\\
009.\quad \textbf{Go to LOOP}\\
\bottomrule
\end{longtable}
~\newline
Because the operation of writing and reading on hash table is within $O(1)$, when consider every leaf needs to query hash table $T$ and $T^{\prime}$ for once and update them for once, then the complexity term is $O(4E)$  for $E=mn$. \\
\newline
\textbf{Exploring.} While the job of graph partition ends, we either have gained an evident geometric layout. We merely need to concern these leaves in the fore region for any root in addition to source, because if there is a visitor wants to reach target $v$ in the $i$th region from source $u$, then he absolutely needs to pass through the therein certain one leaf in the fore region. If we stand at the target to view this progress, we will certainly select amid fore-leaf which might take the total weight minimum for target. Along this idea, in addition to source, by this same strategy we may screen those leaves for their own roots from the second region iteratively till to the target and its own region. Thus the progress may be viewed as an iterative process to compute the minimum total weight for every path form source to those endpoints. Here we can write the form to present this method.

\begin{equation}\label{2}
\begin{split}
w_{u\in r_{j}}=\bigcup_{\text{min}}&(f(x,u) + w_{x})\\
&\text{for } f(x,u)\in W_{\in\tau}\text{ and }x\in (L_{\in\beta}(u)\cap r_{j-1}).
\end{split}
\end{equation}
~\newline
The algorithm is in following.

\renewcommand\arraystretch{1.1}
\begin{longtable}{p{120mm}}
\caption{\textbf{Hybrid Algorithm.}}\\
\toprule
\textbf{Input} \\
$1.\quad \tau\colon \vert \tau\vert =n$, $\tau[i] = L_{\in\beta}(i)$ for $i> 0$ and $i\in V$;\\
$2.\quad W_{\in\tau}=(w_{i})\colon \vert w_{i}\vert = m$ and $w_{i}[j]>0$, it is weight table\\
$3.\quad R=(y_{i})\colon y_{1}=s$, it is graph partition\\
$4.\quad T_{\in V}[i]=t_{i}\colon 1\leq\vert T_{\in V}\vert \leq n$, node $\rightarrow$ subscript of region \\
\newline
\textbf{Output} $\Sigma, P\colon \vert \Sigma\vert = \vert P\vert = n$ and $\Sigma[s]=c$; \\
\midrule
001.\quad\textbf{For} $i=2\rightarrow\vert R\vert$ \textbf{do}\\
002.\quad\quad $u=R[i],w=0,y=0$;\\
003.\quad\quad \textbf{For} $\forall x\in L_{\in\beta}(u)$ \textbf{do}\\
004.\quad\quad\quad \textbf{If} $T[x]<T[u]$ \textbf{and}\\
 \quad\quad\quad\quad\quad\quad\quad\quad\quad($y=0$ \textbf{or} ($y\neq 0$ \textbf{ and }$w>\Sigma[x] + W[x][u]$))\\
005.\quad\quad \quad\textbf{Than} $w=\Sigma[x] + W[x][u],~y=x$;\\
006.\quad \quad\textbf{End}.   \\    
007.\quad$\Sigma[u]=w,P[u]=y$;\\            
008.\quad\textbf{End}.\\ 
\bottomrule
\end{longtable}
~\newline
The program screens each member in any region without concerning their ordinal in the region, and the ergodicity was along on the region queue in partition. In this progress, we used the \emph{visiting structure} to gain leaves and qualify them by querying hash table $T$ to guarantee the each leaf we interested had been in the fore region. Thus we employed a competitive mechanism to pick up the winner, which could make total weight on the root minimum.

The algorithm is to check $n$ leaf-sets for precisely once. Moreover we merely want to seek a leaf amid $m$ members, thus we herein use the method of bubble sorting for only $m$ times comparison. Hence the complexity of algorithm is $O(mn)$. But we must consider there is a query to the weight table $W$ for $mn$ times within whole progress. The weight table is a 2 dimensions array and it is a hash querying to gain a volume array by natural number index which presenting a root. Thus the complexity of query weight table is $cm$ for once in the worst case, so the overall runtime complexity is $O(cmE)$  for $E=mn$. 

The following table exhibits those outcomes of practical runtime from our trials, and the physics unit for measure outcome is $microsecond$ in addition to the option \textquoteleft$k=$\textquoteright, which is the variable about the quantity of endpoints on instance.\footnote {All programs were encoded by \emph{PHP5.0} and running on a laptop with i3-3217U 4 cores, 4G memory and Windows 8.1 OS}

\renewcommand\arraystretch{1.2}
\begin{longtable}{|p{8mm}|p{10mm}|p{10mm}|p{10mm}||p{8mm}|p{14mm}|p{14mm}|p{14mm}|}
\caption{\textbf{Practical Runtime on a Series of Objects}}\\
\toprule
$k=$&\textbf{\emph{G.P.}}&\textbf{\emph{ E.}}&\textbf{\emph{T.T.}}&$k=$&\textbf{\emph{G.P.}}&\textbf{\emph{ E.}}&\textbf{\emph{T.T.}}\\
\midrule
\textbf{20}&3.21&4.45&7.66&\textbf{120}&99.13&163.73&262.86\\
\textbf{40}&9.77&15.37&25.14&\textbf{140}&132.57&223.95&356.52\\
\textbf{60}&21.57&34.77&56.34&\textbf{160}&180.15&301.155&481.3\\
\textbf{80}&45.55&67.57&113.12&\textbf{180}&231.44&384.49&615.93\\
\textbf{100}&68.51&115.09&183.6&\textbf{200}&274.44&474.58&749.02\\
\bottomrule
\end{longtable}
~\newline
The options meaning are: 
\begin{flushleft}
\begin{enumerate} 
\item $G.P.$ is the practical runtime for graph partition.
\item $E.$ is the practical runtime for exploring model. 
\item The final $T.T.$ is total number of the fore two items.
\end{enumerate}
\end{flushleft}
~ \newline
We plotted a curve diagram for an intuitive exhibition to present the relation among the outcomes and relevance instances in above table.

\begin{center}
\includegraphics[height=40mm]{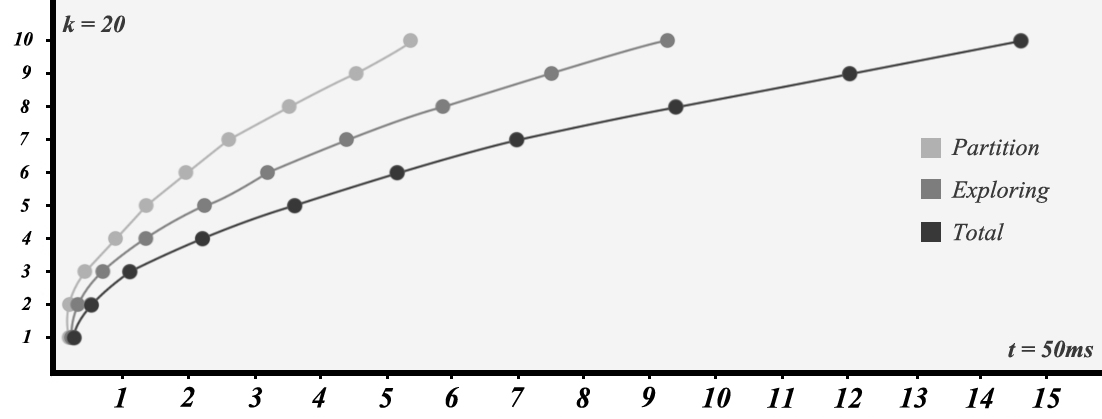}\\
Figure 2
\end{center}
~ \newline
For the experiment object, we chose the square grid figure and the source we specified was at a corner of instance, and for the target we selected the one on the opposite corner to source at the diagonal. We let the variable $k$ as the number of many endpoints on a row or column, which on above diagram the variable range lays on the vertical. Then we had not only the quantity of endpoints on instance with $n=k^2$, but also the quantity of arcs might be $4k(k-1)$. In this trial, the $k$ value is in the range of 20 to 200, so that it may be properly to say that they are representing these actual entries which were in an array of square numbers $(400, 1600,\ldots, 4\cdot 10^4)$. Hence to the curve, however to partition or exploring, they both showcased the variance of quadratic form.
 
For example, the greatest number on the vertical is 10, which it means the instance is containing 40 000 nodes and 159 200 arcs. The scale might be enough to cover any biggest city on the current earth. I think you may learn what significance these numbers implied in practice to those ventures, which items are about to vehicle navigation system.

\section{Static System}
To above algorithm, we might say it is a hybrid algorithm, and the new approach only needs the features of finite for any instance as entry. Instead of stating the connected attribute, because we might further use the graph partition to qualify the instance. Now you can find more similar methods to other methods in this new approach, in special to \emph{Dijkstra\textquoteright s algorithm}\cite{3}, \emph{Bellmanford\textquoteright s algorithm}\cite{4}, \emph{Floyd\textquoteright s algorithm}\cite{5}, \emph{Johnson\textquoteright s algorithm}\cite{6} and etc. Since the layout produced by graph partition, we might have a qualitative analysis about those contexts for these pioneers, namely we can comment them in a geometric way.

To use Dijkstra\textquoteright s algorithm, there someone definitely suggest you to prepare a perfect directed graph with a few of endpoints, because the method could not harness the visiting to nodes and perhaps turn back to source or repeated segment on a \emph{mixed} figure. To Bellmanford\textquoteright s algorithm, we must either prepare a group of rather crafted numbers for endpoints\textquoteright~labels, which we hope this namespace maybe correctly express some inside certain relationships, so that program could turn out an optimal result. In other word, the calculation is build on a uncertain basis framework, hence there someone maybe warn you that the entry figure has not got complex relationship, but he cannot answer what is a complex relationship. Here we issue the most representative contexts to express a common notion that these vague contexts actually are from a same source, the geometric layout, through they have had so much distinct before. If you study Bellmanford’s method on a small sample, you can find that the crafted numbers label just achieve such idea which exploring along the level like on a tree structure.

The above graph partition approach provides a homogenous figure to the instance, so that the intuitive geometry layout might properly emerge. If we set the new figure as a network in a physics field, and let the potential stream flows from source, passes through the interval regions and finally stops at the target and its own region, we might naturally have each region regarded as a nodes cluster at the contour line and its subscript likes to label the energy gradient, the shortest path to target either may become the shortest path of physics stream.

The above notion is the key to comment these prior approaches, but we cannot deny them, that we have ever been attempting to seek out the geometric meaning as well, till the field layout appeared. So in this section, we prime and thoroughly clarify these conceptions at logical level, and their relationships to the new approach that we made. Looking at the following table, which is a trial on a \emph{simple} grid figure with $n=5^2$, at where source and target are both on the opposite corners at the diagonal line.

\renewcommand\arraystretch{1.0}
\begin{longtable}{|p{10mm}|p{22mm}|p{22mm}|p{22mm}| p{22mm}|}
\caption{\textbf{Outcomes of BOTS and EBOGP.}}\\
\toprule
\textbf{No.}&\textbf{\emph{Algorithm}}&\textbf{\emph{Results}}&\textbf{\emph{Loops}}&\textbf{\emph{Time(ms)}}\\		
\midrule
$\textbf{1}$&$BOTS$&8\,512&81\,598&5\,000\\
\midrule
$\textbf{2}$&$EBOGP$&$70$&$252$&$5.53$\\
\bottomrule
\end{longtable}
~\\ \newline
Look at the second row in the above table, the job of \emph{BOTS}\cite{1} algorithm actually is used to search $Hamilton$ path, which yet could enumerate all possible connected paths among pairs source and target entirely. In the paper \cite{1}, we have proved that each path $p$ therein outcomes has $1\leq\vert p\vert\leq n$ , if the entry is simple. To other hand, by the Dijkstra\textquoteright s method we actually seek the extreme total weight, nor concern the length, as if it seems to search all possible and find out the only one ideal result like BOTS method.

Even if the accuracy is strict, but for some numbers in above table, it issues the accuracy is not whole story, nor even a critical role, because of another factor about to computing resources. That is amazing for a small number $n=25$ as a little entry, it has turned out a outcome, a large number which there are 8\,512 paths. That is caused by combinative possibility, likewise this is why you maybe ever before receive the warning about the scale of entry for using Dijkstra\textquoteright s algorithm, the warning was done since that what only a small entry might quickly exhaust your entire computing resource, often showing the case out of memory, and even not a method for you to forewarn the crash.

For the second algorithm \emph{EBOGP}\cite{2}, that author briefly names it here. This approach does the exploring job based on the layout of graph partition, that has been introduced in the paper\cite{2}. Nevertheless we can learn the scale of outcome is only less than 10\% to BOTS with $k=5$, but the complexity is attributed to the factorial item about $\sqrt{n}$. For a square grid instance, the form $S(2k-2,k-1)$\cite{2} exactly presents the relationship between outcome and entry, and the quantity of loops in program might be $S(2k,k)$. Namely, it said that growth still grips enough power to collapse any computer with $k$ just being a bit large number.
 
We denote the result set from EBOGP by $R_{EBOGP}$, that obviously it is the subset of BOTS ones, which we denote by $R_{BOTS}$. The set $R_{EBOGP}$ is particular to what it is the collection of all shortest paths from source to target, for any instance. Therefore we have an idea to exam our new algorithm, herein we call the new method \emph{Hybrid Dijkstra\textquoteright s algorithm}, that we might have the abstracted sequence $R_{HDA}\subseteq R_{EBOGP}\subseteq R_{BOTS}$  to discuss the problem about the accuracy of outcomes yielded by these approaches.

Before our discussion, we must mention a major and critical conception \emph{triangle inequality}\cite{7}. We popularize this inception to polygon, that there is a definition as follow:

\begin{definition}
Let A be a collection $A=\{x_{i}\}_{i=1}^{n\geq 3}$ and $\forall x_{i}>0$. There is a natural number $M$ with $3\leq M\leq n$. Let $M=m_1+ m_2$ with $m_1,m_2 > 0$. If there are two subsets $a_1,a_2\subseteq A$ with $m_1=\vert a_1\vert$ and $m_2=\vert a_2\vert$, we call they satisfy \emph{polygon inequality relationship}, if and only if 
\begin{center}$m_1\leq m_2\Longleftrightarrow \sum _{\in a_1}x_{i}\leq \sum_{\in a_2}x_{j}.$\end{center}
 Otherwise we call \emph{anti polygon inequality relationship} for $\sum_{\in a_1}x_{i}>\sum_{\in a_2}x_{j}$.
\end{definition}
~\newline
For our notion, the inequality relationship depictures the fact that we assume there is an instance of polygon with a source and a target on it, if we find the shortest path from source to target, then it similarly might be the optimal one based on the polygon inequality relationship. By this rule support, we can deem the optimal path amid those members in set $R_{EBOGP}$ either is that global one on instance, hence we may focus on the problem that how the outcome of hybrid approach is. Here, we firstly focus on the case for all weights on instance equal to a unique constant, so we might prove the inference as follow.

\begin{inference}
Give a finite weighted instance $G=(V,\tau,W)$. Let the partition $R$ on it with $R=(r_{i})_{i=1}^k$. Consider two endpoints $u, v$ with $u\in r_{1}$ and $\vert r_{1}\vert=1$, let $v$ in another region with $v\in r_{i}$ for $1<i\leq k$. For any shortest path $p$ from the source $u$ to target $v$, such that it is an optimal path from $u$ to $v$, if and only if all the weights on $G$ are equal.
\end{inference}

\begin{proof}
According to the above-given conditions, we let the shortest path $p=(x_{s})_{s=1}^{i}$ with $x_{1}=u$ and $x_{i}=v$. We can conclude its length has $i-1=\vert p\vert_{\in\tau}$ by the shortest path theorem of graph partition\cite{1}. It is obvious to compute a total weight for $p$ by above form\eqref{1}, having $w(p)=\sum_{1\leq s\leq i-1}w_{s}$. We let there is another connected path $p^{\prime}$ from source to target with $\vert p^{\prime}\vert=i^{\prime}\geq i$. It is obvious that if the whole weights are equal to a constant $c$, we can present their relationship like that $c(i^{\prime}-1)\geq c(i - 1)$, or else we can always tick distinct weights to turn out a converse unless there is only path $p$ between them. Hence this inference holds. 

\end{proof}
~\newline
This inference showcases a fact that set $R_{EBOGP}$ not only is the collection of shortest path and, but for likewise it is the collection of all optimal paths too, if and only if whole weights are an identical constant. Furthermore the feature implies that the perfect result most possibly lays in the set $R_{EBOGP}$ if the polygon inequality relationship is a popular law in figure. The case in this inference may be rather extreme, so we can regard it as a boundary on the $density$ and $cardinality$, as two key criteria to be the characteristics of data structure of those weights in a complex system. For example the density of all weights is so tiny that they are approaching to equal, thus the density may decide the outcome in a sense way, beside cardinality. 

Meanwhile beside the Bellmanford\textquoteright s algorithm, its exploring is just stage by stage along the levels in instance likes hybrid approach does among regions. Certainly the Bellmanford\textquoteright s algorithm need a group of crafted labels to tick those endpoints to guarantee operation does at the correct orbits, so as to this request complicates the job of encoding and it is hard to seek out the inside and fatal logical error. But for hybrid method, this burdensome job has been dismissed. This is the reason that why author called this new approach hybrid, which it is a synthesis of those advance ideas in pioneers, any way to avoid their various drawbacks since clear geometry layout.

Any way we need have further to prove the statement that this hybrid method could produce a path and it would be the optimal result in collection $R_{EBOGP}$ if and only if entire weights are of fixed and the their numbers are arbitrary\footnote{Please note that: here is not to say the result is the optimal to set $R_{BOTS}$. About this issue, we will sound it in later section}.

\begin{theorem}
Give a finite weighted instance $G=(V,\tau,W)$. We let a graph partition $R=(r_{i})_{i=1}^k$ on it. Consider an endpoint $u$ being a source in the first region which is a singleton one. Let an endpoint $v$ in other $i$\emph{th} region as target. If the hybrid approach is used to explore the path among $u$ and $v$, then the method merely yields a shortest path among them.
\end{theorem}

\begin{proof}
Link to the above precondition, for the hybrid algorithm, the operation might pick up a leaf $x\in L_{\in\beta}(v)$ for root $v$ underlying visiting structure $\beta$, if and only if the leaf make the total weight $w_{v}$ at root $v$ is less than other leaves. In the above program, we let an array P to gather the leaf $x$ for root $v$ as output, which data structure is presenting P$[v]=x$. The data structure shows a logical relation: there is only an arc link to $v$ from $x$, so that we are certainly ongoing to conclude the fact that the endpoint $x$ similarly has iteratively got an arc reflect to itself from a fore-leaf, till to end at the source $u$ by the precondition for first region is the singleton. Meanwhile we can counter the length of the tracing stage is $i-1$. So this tracing process with a recursive feature has a virtual inverse link to source $u$ from target $v$, at where we can assume that the virtual link is connected to an endpoint $t$ in $j$th region with $j>1$ and stops there. As well when we query root $t$ for its own fore-leaf in the array P, we might have the contradiction that there is a fore-leaf for root $t$ stay in $(j-1)$th region recording in array P, not coincide to our above assumption. Hence we prove therein outcome, at least there is a path connect target from source. Actually we also prove all endpoints are targets in addition to source, certainly the source either may be the empty-node\textquoteright s target. 

We assume a case again that there is another path $p^{\prime}$ from source to target in the cutting graph which has been depictured by array P.  Any way the case means either target or a medium endpoint should be the common node on fork about two paths $p\text{ and }p^{\prime}$, so that it always lets this common node have two fore-leaves in array P, and contradict, hence the case does not exist. We prove that underlying this data structure amid outcomes, the hybrid method just produces a unique result for a given pairs source and target.

We can let a path length greater than $i-1$, which links to target by the later-leaf $y$. It is obvious that it violates the logical structure of P. Likewise to the endpoint $y$ as a root in array P, which the amount of fore-leaf is only equal to 1. Hence there is no such path in cutting graph and we prove the only one result has to be the shortest path, this theorem holds.

\end{proof}
~\newline
By the above theorem, we can learn the fact that the data array P figures out the new cutting figure as a tree, where the tree root is source \footnote{But it is not a minimum span tree, that need you note.}. Any way, we have proved the notion  $R_{HDA}\subseteq R_{EBOGP}$, that the hybrid approach can make only one shortest path for a source to whole others. So we have to answer the question: whether the result made by hybrid algorithm is the optimal one among those opponents in set $R_{EBOGP}$. The conclusion and its proof are in following.

\begin{theorem}
Give a finite weighted instance $G=(V,\tau,W)$. We let a graph partition $R=(r_{i})_{i=1}^k$ on it. Consider an endpoint $u$ being a source in the first region which is a singleton one. Let an endpoint $v$ in other $i$\emph{th} region. If the shortest path collection $R_{EBOGP} (u,v)$ is non-empty, then the result $p$ made by hybrid approach has the minimum total weight in it.
\end{theorem}

\begin{proof}
To those above-conditions, we let path $p$ be made by hybrid approach from $u$ to $v$. It is obvious for the conclusion to naturally hold, if and only if the cardinality of set  $R_{EBOGP} (u,v)$ is only one. So we let it greater than 1. Consider the target in the second region, that it has not any choice but for source as only one fore-leaf. Hence the conclusion might hold in the fore context. 

Consider there are $i\geq 2$, and $\vert R_{EBOGP} (u,v)\vert \geq 2$, By the competitive mechanism described by form\eqref{2}, for those endpoints in $3$rd region, we can prove their total weight recorded in array $\Sigma$ are minimum among their own fore-leaves respectively. It is obvious to employ the \emph{induction principle} to prove the case similarly appears at target $v$ iteratively.

By theorem 1 we prove there is only one shortest path between source and target, which lays in array P, hence this theorem holds.  

\end{proof}
~\newline
This theorem 2 eventually proves the new hybrid method is an exact solution to reveal optimal path for static weighted figure under the wing of polygon inequality. But in the real world, always there is someone to draw out a path on instance to make a counter example. Likewise, this also is not a hard thing for you to draw a \emph{Hamiltonian path} and merely tick minimum weights on each segment to debate this method.

This absolutely boring path issues a question: if we regard a complex system as that there a mechanism stochastically assigns some numbers to whole arcs, then we might along this idea to study this complex system how about the stochastic probability to anti polygon inequality, we here call it \emph{anti event} simply.

So that we always might imagine a complex system like that there has a number array used to assign its own members to a figure. At the intuitive notion, we at least deem the anti event is involving to the number array and those weights distribution on the figure. For to clarify the problem, we issue an example.
\\~\newline
\textbf{Example.} Given a natural number array $A=\{a_{i}\}_{i=1}^{n}\colon a_{i}=i$ for assignment and consider a polygon $\Delta =(e_{i})_{i=1}^{t\geq 3}$. We firstly let $t=3$, so that there are these operations: while we select three numbers $x_1,x_2,y$ from array $A$ and assign them to those edges $e_1,e_2,e_3$ respectively on polygon $\Delta$ as weights, if there is $x_1+x_2< y$ , we would deem the anti event inevitably happens, here we regard the edge $e_3$ with number $y$ as the shortest path.

For the variable $y$, its domain has to be in $[2,n]$, due to the max number in array $A$ is $n$. And since that, beside the two other items on left side of above inequality, such that the number $n/2$ is a label for their domains, which means they might not both likewise equal to and greater than this label; if so, their sum would greater than $n$ so as to the anti event should never appear. We consider the case $x_1\neq x_2$ primarily. Along the fore line, variable $x_1$ has such domain $x_1=[1,(n-1)/2]$, since that followed there may be domains and relationships among three variables respectively in following.
\begin{align*}
 x_1=i&\Longrightarrow x_2=[i+1,n-i-1] &\text{ for }i=[1,n/2]\\
~&\Longrightarrow y=[i+x_2+1,n].\quad &\quad
\end{align*}
~\newline
Then the total possibility for $x_1$  equal to a certain volume $i$ will be

\begin{align*}
p(i)&=\sum_{x=i+1}^{N-1}n-(i+x+1)+1 &\text{for }N=n-i;\\
&=\sum_{x=i+1}^{N-1}N-x. &\quad
\end{align*}
~\newline
Therefore the sum of total possibility for $i=[1,n/2]$ is

\[P=\sum_ {i=1}^{T}p(i)=\sum_{i=1}^{T}\sum_{x=i+1}^{N-1}N-x\quad\text{for }T=(n-1)/2,~N=n-i.\]

~\newline
We have

\begin{align*}
P&=p(1)+p(2)+\cdots+p(T)\\
&=(n-3)+(n-4)+\cdots+1+\\
&\quad ~ (n-5)+(n-6)+\cdots+1+\\
&\quad\quad \vdots\\
&\quad~ +1
\end{align*}

~\newline
Finally the succinct form is

\[P=\frac{ (n-2)(n-3)}{2}+\frac{(n-3)(n-4)}{2}+\cdots+1.\]

~\newline
While the form is divided by combinational form $S(n,3)$ we have

\begin{equation}\label{3}
\begin{split}
\rho &= \frac{P}{C_n^3}=3\left(\frac{(n-2)(n-3)}{n(n-1)(n-2)}+\frac{(n-3)(n-4)}{n(n-1)(n-2)}+\cdots+\frac{1}{P_n^3}\right)\\
~&=\frac{3}{n}\varphi\\
\text{ for }&\varphi=\left(\frac{(n-2)(n-3)}{(n-1)(n-2)}+\frac{(n-3)(n-4)}{(n-1)(n-2)}+\cdots+\frac{1}{P_{n-1}^2}\right).
\end{split}
\end{equation}                             
~\newline
For the above term, we consider the array $\varphi$ is convergence towards to 0 while the variable $n$ toward to infinite, we denote the term by $\rho^3$, so that we can deem the probability on the triangle surely tends to a constant, an extremal value. But we only finish the job to pick up three numbers. But while we plan to assign them to those edges on a triangle $\Delta$, the probability would be $\rho^3/3$ for assignment the greatest number $y$ to segment $e_3$. This example demonstrates that the probability of anti event either involves the quantity of those edges on polygon $\Delta$.

Second for the case $x_1=x_2$ such that having  $x_1+x_2=(2,4,\ldots,k)$ for $k=n-2$ if $n$ is an even number otherwise $k=n-1$. Then with the above same fashion we have
\[P^{\prime}=(n-2)+(n-4)+\cdots+1.\]
Then
\[\rho^{\prime}=\frac{P^{\prime}}{S(n,3)} \approx \frac{3}{n}.\]
~\newline
Of course while the $n$ tends toward to infinity, the probability $\rho^{\prime}$ is toward to 0. Thus here we might only interest the case of $x_1\neq x_2$, which is approaching to $0.48$, but for the case of assigning the greatest number amid three ones to edge $e_3$, that volume either becomes near to $0.17$ that $\rho^3$ divided by $S(3,1)$.

As well the simple conclusion should be in a strict context of choosing numbers from assignment array. It shows the case that the anti event either is not more popular than its opponent. But we have to recognize the fact that there is not any more theory about the probability of anti event. In general, the current theory also refers to the case only for one-to-many model, so that we inevitably suffer an embarrassing affair that we study or profile the problem is merely depending on through the intuitive \emph{sample space}.

We had a simulated trial with several arguments: a natural number array $A$; the letter $E$ presented the amount of edges on a polygon;  our experiment was designed to that we randomly got $E$ numbers from the array $A$ and randomly dispatched them to two groups to have their respective sums and cardinalities compare. According to the definition of anti event, we could identify whether the anti event happened between the two groups data. We would let amid one as a shortest path whose length was represented by its own cardinality. Our testing had its length to another longer path done from 1 to $E/2 - 1$ if $E$ was a even number, otherwise $(E-1)/2$. The testing was repeatedly done for 1\,000 times, where for every time, we counter each anti event on distinct cardinality of the shortest group. Eventually we gained a group of volumes upon various lengths to respective percentages. In practical trial, we initialize them as $A=1,2,\ldots,170$ and $E=169$. We plotted the histogram to show the outcome in following. 

\begin{center}
\includegraphics[height=40mm]{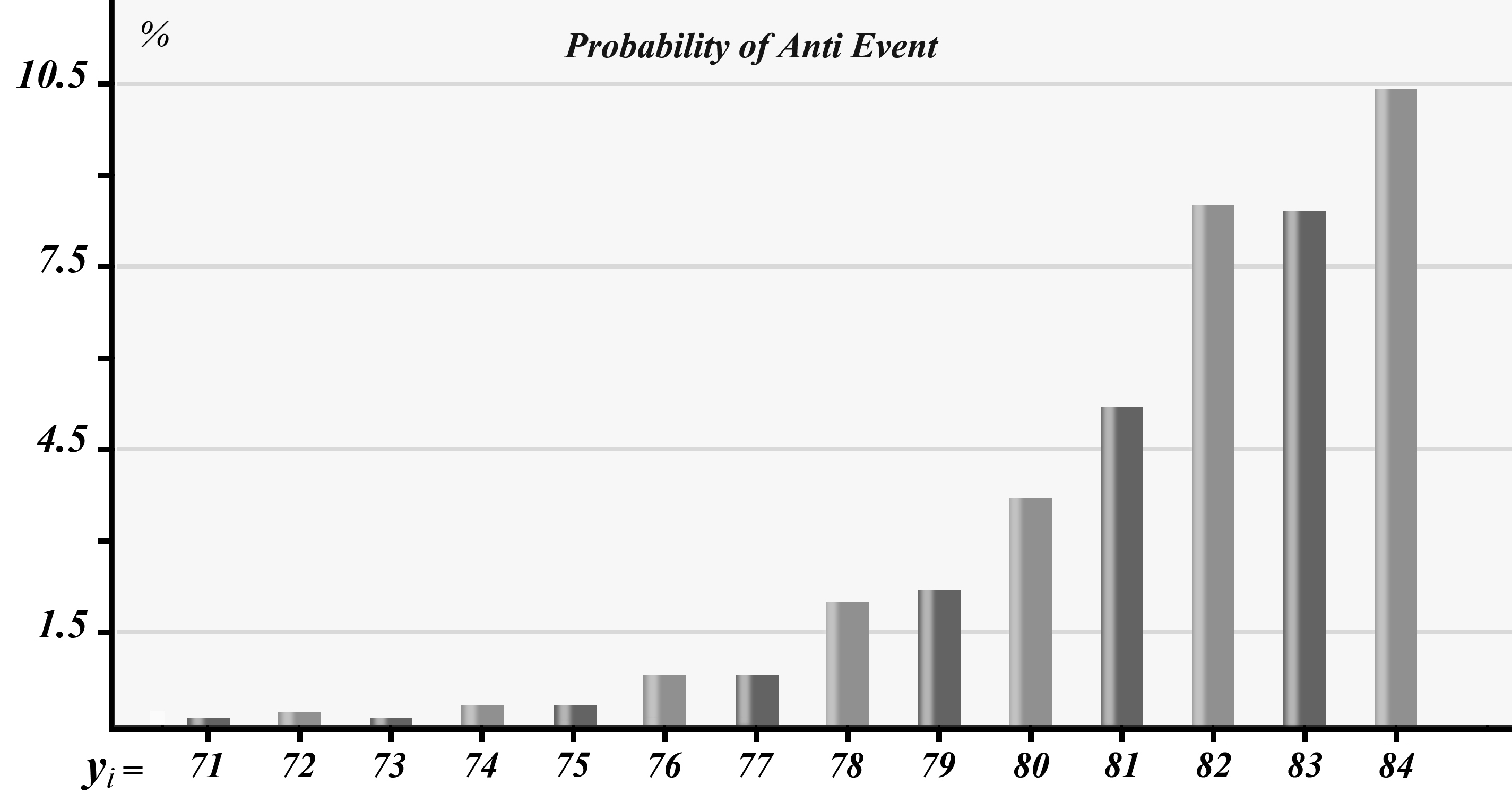}\\
Figure 3: Polygon Inequality Probability
\end{center}
~\newline
For the above diagram, when the length exceeded 70, the anti event began to appear. Namely the anti event is a small probability case while the length of simulated shortest path was less than 70. But for the length being 84 near to half of 169, the probability was near to 11\%. From this trial, we might gain a knowing about what the anti event appears in a high frequency is mostly to involve the case, which two cardinalities of groups are near to each other.

If we imagine the figure in practice is constructed by many distinct polygons, we might deem the fact that to a shortest path made by hybrid method, the anti event easily happens due to these paths with an approaching length. Thus we have to interest the severity about the influence of anti event to accuracy and seek a way to solve it, else in the application, the top probability of anti event in this trial is 10.4\% less than the triangle with three edges. So we want to learn the relationship among the probability of anti event to the scale increasing on a weighted figure. In following portion, we will have along on the tendency of the fore idea to do another experiment.\\
\newline
\textbf{Greedy Idea Trials.} We designed a trial with a seriese objects by using greedy idea attempt to remove the anti event. We chose a group of grid instances as objects for our trial, that is $k=5,10,15,\ldots,50$, which endpoints quantity $n$ was in range $[25,2500]$ as to form $n=k^2$. We let these members in array $A$ be randomly assigned to each arc on figure as a fixed weight. We might let $A=(1,2,3,\ldots,k)$ , certainly the cardinality of array $A$ was $k$ and, for this assignment operation we call \emph{regular assignment}. 

We designed an algorithm that we sorted those arcs\textquoteright~weights as a queue $w$ in the ascent. Then after, we initialized an empty instance $g$ without any arc or endpoint. We took some arcs off for once to draw on the instance $g$, which their weights are in the front range of positions in array $w$. This progress was iteratively to cut the original figure to build up the new figure $g$ again and again till a condition was satisfied after we finished drawing a strain arcs: the new partition $\sigma\subseteq g$ contains the source and target both. Of course the case presents source connect with the target. Then we performed the hybrid method to search the optimal path for source and target. Through this method, we wanted to screen some arcs out with less weight to piece up a proper figure, that guided by greed idea, we always deem the optimal path should be organized by those $nice$ arcs and in most probability, it is properly laying in such instance likes $\sigma$. 

\begin{center}
\includegraphics[height=25mm]{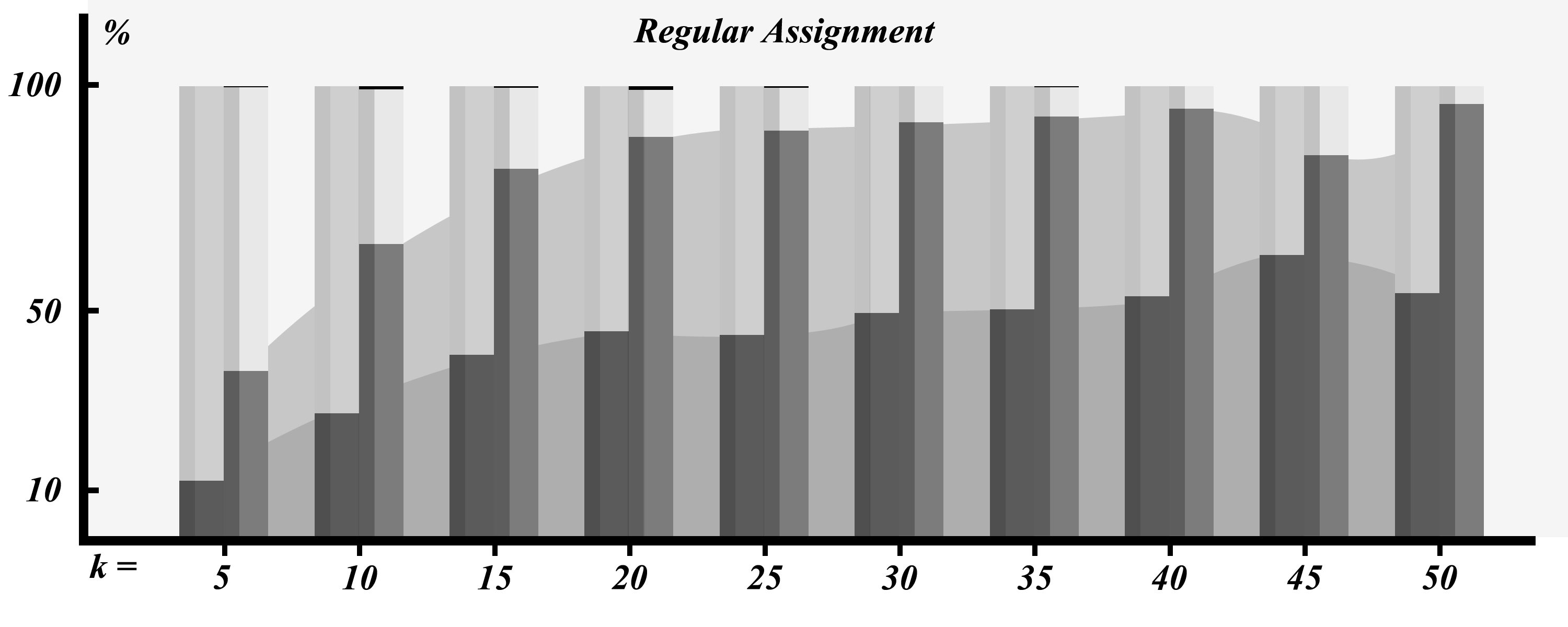}
\includegraphics[height=25mm]{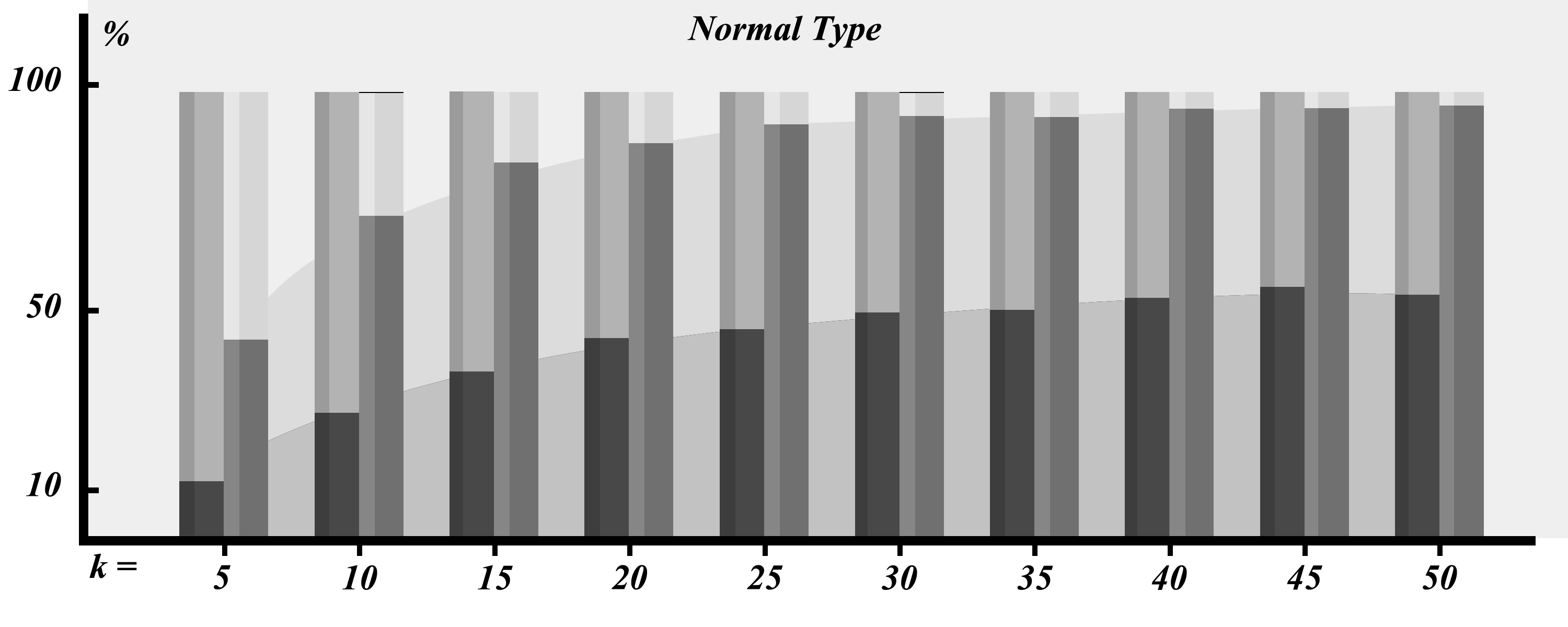}\\
Figure 4th is regular, 5th is nornal.
\end{center}
~\newline
We compared the outcomes produced by the two approaches in regular assignment. We briefly call the new approach \emph{greedy hybrid approach}. The above histogram visually demonstrates the relationships among those outcomes with those gray pillars. For a pair pillars on each number at the horizon, the \emph{left} pillar with deep-high light grades is showing the contrast about two paths\textquoteright~lengths made by two approaches respectively, which the deep gray area means the rating of greedy one longer than the hybrids, or else the high light area says the equal relation. This expression either suits to the right pillar for about the comparison of two total weights at target. But yet note the tiny black area, it is showing the less than relation, where the greedy one is winner. The outcomes were upon testing 1\,000 times for each volume on horizon.

The trial issues the case that although we strived to filter some greater weights, but the utility did not mend even more, which on above diagram the rating changing has shown the outcome off. This solution depending on through each arc to restrain the anti event happen is obvious in failure, at least said on grid figure. Maybe for the cardinality of set $R_{EBOGP}$ could be $S(2k-2,k-1)$ for variable $k$, the outcomes present a factorial increasing might take the probability of anti event fall down. At least author deems through this trial, this case shows the anti event as a critical key impact to take a decision in failure on employing greed method.  

\begin{center}
\includegraphics[height=25mm]{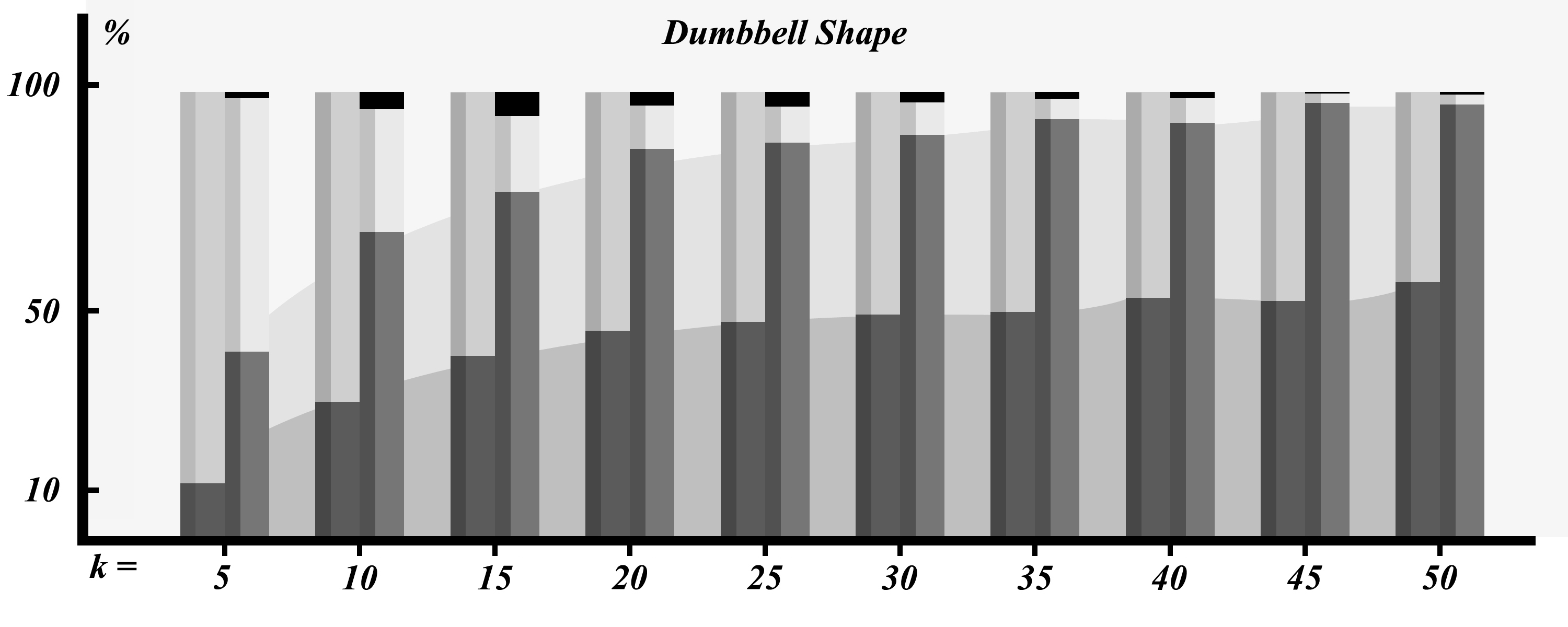}
\includegraphics[height=25mm]{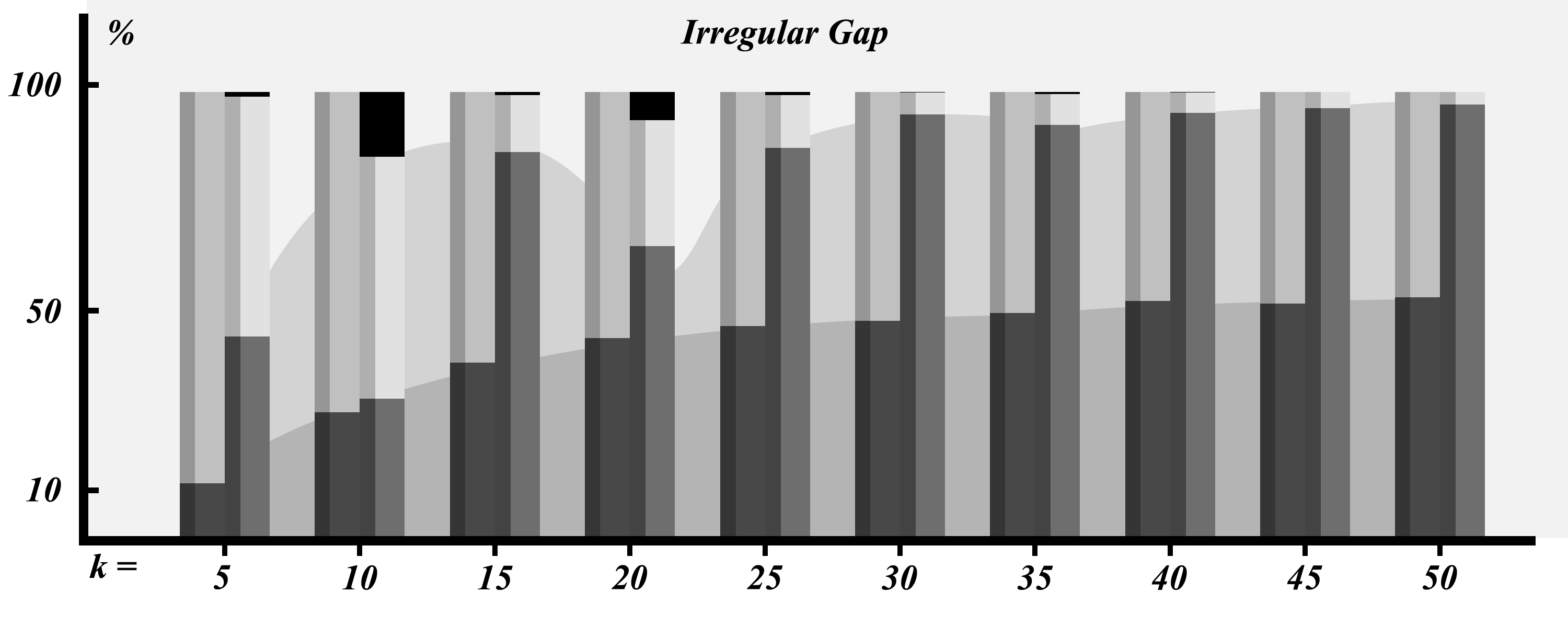}\\
Figure 6,7 
\end{center}

~\newline
We had continued to do some experiments by altering density in assignment array. They  respectively were normal type, dumbbell-shape type and the gap being random. Those histograms are listed above, and they are commonly showing the greedy approach cannot be in leading to the hybrid one, however on length or accuracy. Eventually we can say using greedy method on the side of entry might mend the accuracy in a small probability.\\
\newline
\textbf{Summary.} We played our strategy on two stages, and eventually we gained a unique outcome which is not only shortest path but also an optimal one among those members in set $R_{EBOGP}$. We might issue the strict context of this method, that under the circumstance of with least cardinality of nodes, the method can achieves an ergodicity to satisfy an extreme condition of weighted request. Here we do not mention which an extreme condition. In fact, if we only alter the operation sign in above program, we can reform the exploring job to search the maximum total weight. In other words, we hold the selection criteria so as to the word optimal has more popular meanings. We ever use the greedy method in our solution, but we only seek out the drawback of greed idea in our strategy. About the problem of seeking out an extreme accuracy, we will issue the involving contents in later section.

Since that, through our above conclusion, does it mean the EBOGP approach might be ignored? Here we only discussed the case about fixed weighted figure, although there numbers are randomly dispatched to each arc. If we use the hybrid approach in a dynamical system, we would suffer the problem that how to exam the only one result is what, even to determining whether it is an acceptable volume or not. At least author deems this is a massive embarrassment that the prior methods suffered. Eventually they had not any agency to trial and avail their results, so as to there were any more choices, and they had tragic to stop at the static system for decades.

\section{Exploing in Dynamical System}
Dynamical system\cite{8} is an indefinite thing hard to be figured out by natural language or mathematical language entirely. Here we will not dabble in much more at the hand of definition or attribute. We will define a simple dynamical system, that for each weighted arc $(u,v)$,  on it the weight is determined by the form $w(u,v)=f(v,I)$, which means any weight on each arc is produced by function, no longer in fixed status and there are many possible values. So we can design the variable $I$ might be the sum of each fixed index on respective arc. At least if we use a group numbers to assign these index and weight table randomly, we might make such uncertain factor for each weight: although the volumes of index and total weight will increase along with the path length growth. But on the concrete case of querying weight table for a weight with an arc as amid entry, to a target and its own fore-leaf, although it is same leaf, but for some distinct paths stop there, such that we cannot confirm the queried outcomes equal or unequal about their respective total index.

\begin{definition}
For the total weight on endpoint t might be recursively claimed by following forms.
\begin{flushleft}
\begin{enumerate} 
\item $I_{t}=c_{1}, w_{t}=c_{2}$; \emph{if $t$ is the source, the index and weight are constants.}
\item $I_{t}=I_{s} + I_{(s,t)}$ \emph{for } $I_{(s,t)}\in D_{\in\tau}$ \emph{and} $(s,t)\in\tau$;
\item $w_{t}=w_{x} + f(t,I_t )$ \emph{for} $x\in L_{\in\beta}(t)$ \emph{and} $f\in W$;
\end{enumerate}
\end{flushleft}
\end{definition}
~\newline
This definition demonstrates how to gain a dynamical weight on any arc. Here actually we just alter the name and utility of weight table which was in above static system, instead of storage the fixed indices for each arc. The progress merely adds an operation for query the weight table than in static system. So the runtime complexity about method should be unchanged.

Beside the hybrid approach, while we have the shortest path set $R_{EBOGP}$, we can compute all paths in it and seek out the optimal one. Instead, we cannot do that because there may have a request about computed resource at factorial increasing. On the other hand, at least on recent, the grid figure is the only one whom we might learn the amount of shortest path for given any pair of source and target on this type instance, through combinational theory. Not also mention those weird instances, we definitely cannot value the scale of outcome by a certain method before start our computed task. That is why author always used this type figure, although the EBOGP is a naive approach. Therefore at least on a small sample, we can utilize trials to exam the hybrid approach and attempt to reveal some interior features on the complex system, that sample space also is a critical importance mathematical method. This is just the reason that the accuracy made by EBOGP has absolutely exact, so that we have comparative trial about upon two methods to study algorithms and complex system.\\
\newline
\textbf{Experiment Design.} We let two number arrays $d$ and $w$ as assignment array for index table $D$ and weight table $W$ respectively, which all numbers in them are positive greater than 0. We use two arrays to do stochastic assignment job. According to the above definition about dynamic, we design a plan for trial: at first we will test the hybrid approach in such dynamical system, and then after to perform the EBOGP approach to produce the set $R_{EBOGP}$ and for whole shortest paths to cast their own respective total weights. As well we would sort these outcomes as ascent and remove those repeated numbers out, so that we may produce an ordinal sequence $X$ with $X=(x_{i})_{i=1}^N$ for $i<j\Rightarrow x_{i} <x_{j}$. Then we can observe the rating of hybrid result in set $X$, the position in sequence $X$. At the next stage, what we need further to do is recording the information about those ratings for study accuracy and system. There are several arguments in following, and they need to be in our attention.

\begin{flushleft}
\begin{enumerate} 
\item The variable $k$, here we let $k=10$ and fix it.
\item Let assignment array $d=\{d_i \}_{i=1}^{k}\colon d_{i}=i$.
\item The assignment array $w=\{w_i\}_{i=1}^N\colon w_{i}=i$ has a \emph{variance} cardinality. 
\item The $1^{\emph{st}}$ outcome is the \emph{top number}, which is the statistic percentage of hybrid result being on the $1^{\emph{st}}$ position after T times testing.
\item The 2nd outcome is the \emph{worst position} to label the greatest position therein T results.
\item The 3rd outcome is the \emph{average status}. After T times testing, we can gather two groups data. We let the first group $\alpha$ contains T volumes of position made by hybrid method, likewise to second group $\beta$, each member is the amount of sorted total weights made by EBOGPs for each testing. Therefore we can have the ratio about their average as $\delta = \overline{\alpha}/\overline{\beta}$.
\end{enumerate}
\end{flushleft}
~\newline
We plotted those arguments in following diagram.

\begin{center}
\includegraphics[height=50mm]{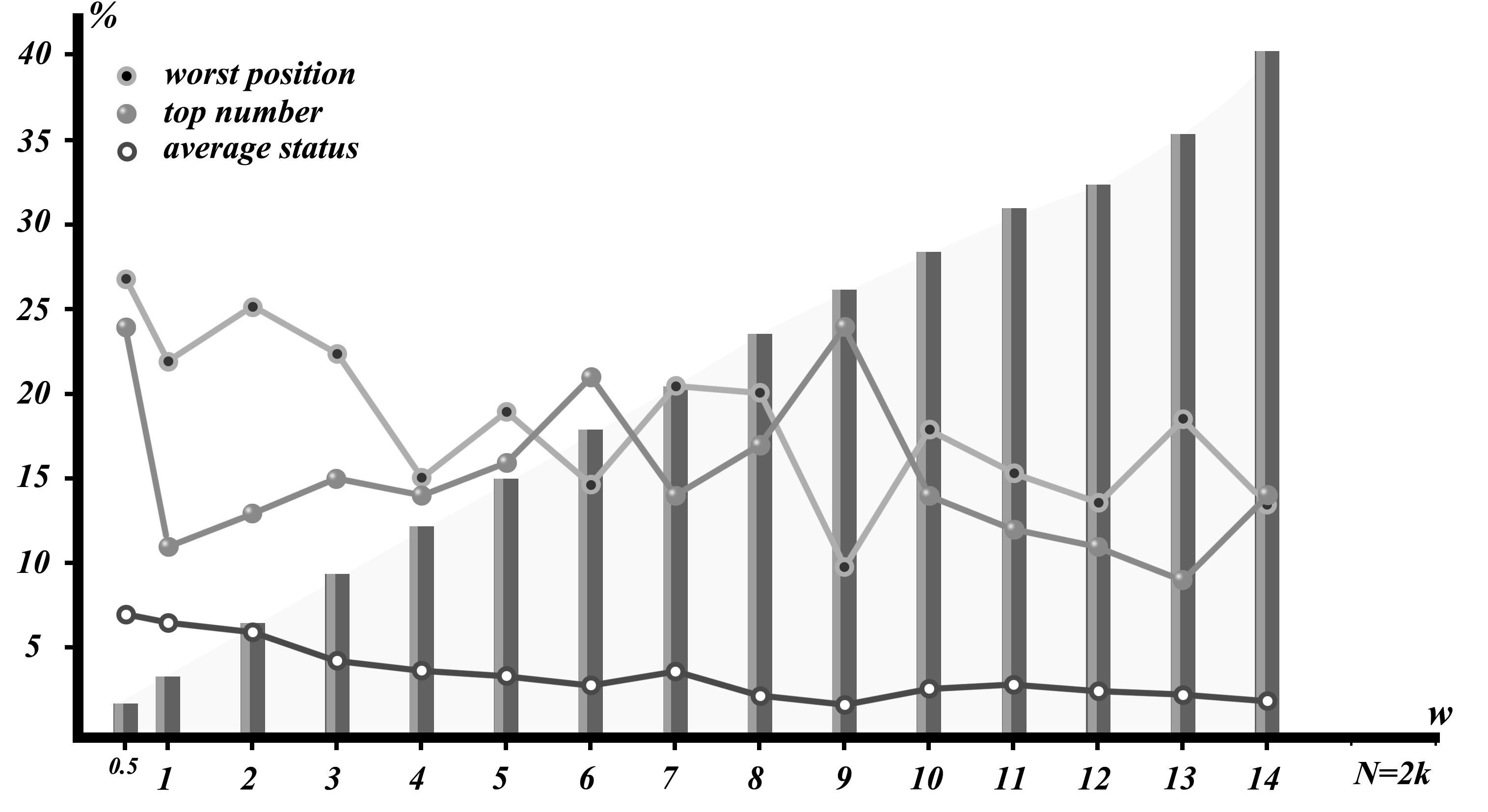}\\
Figure 8
\end{center}
~\newline
On the $x$-axis, a unit is $2k$ to present the rise of cardinality for array $w$. And these gray pillars show the $\overline{\beta}$ values upon distinct cardinality of array $w$, indeed, we might observe it increases according to linear along with the number on horizon. Instead, for three broken lines of worst position, top number and average status, as if there seems to have no more relationships among them. But here are some interesting things enough to encourage us: for that the worst position, at least it is less than 30\%. For general score system, we often use 5 rating to assess an case within \emph{excellent}, \emph{good}, \emph{ordinary}, \emph{bad} and \emph{worse}. There in the worst case the outcome is still to keep in good degree. Either to those members in $\alpha$ group, the data distribution was not a normal type as well, which the quantity of occupying the $1^{\emph{st}}$ position often was the most amid outcomes. Note the average status line presents a moderately drooping to horizon, for the case we can deem it represents a tendency that for the cardinality of weights in system toward to a large number maybe compensate certain drawback in hybrid progress, or the accuracy maybe ongoing in improvement along with weight cardinality rise. Of course, these notions need much more evidences to support. 

~\newline
\textbf{Change Density.} In above trials, we used a natural number array $w$ with the gap among neighbors being same and identical to 1, which we can say the density about array $w$ be \emph{even}.  As well we changed the density and let it greater than 1 and continued to test this case, so that we can further observe the gap how to influence the outcomes.

Herein we let array $w$ as $w=(w_{i})_{i=1}^{N}\colon w_{i}=c\cdot i$. We employed four distinct gaps on it and let $c$ be 1,2,3,4 respectively and only drew their average status lines. The diagram is in following.

\begin{center}
\includegraphics[height=45mm]{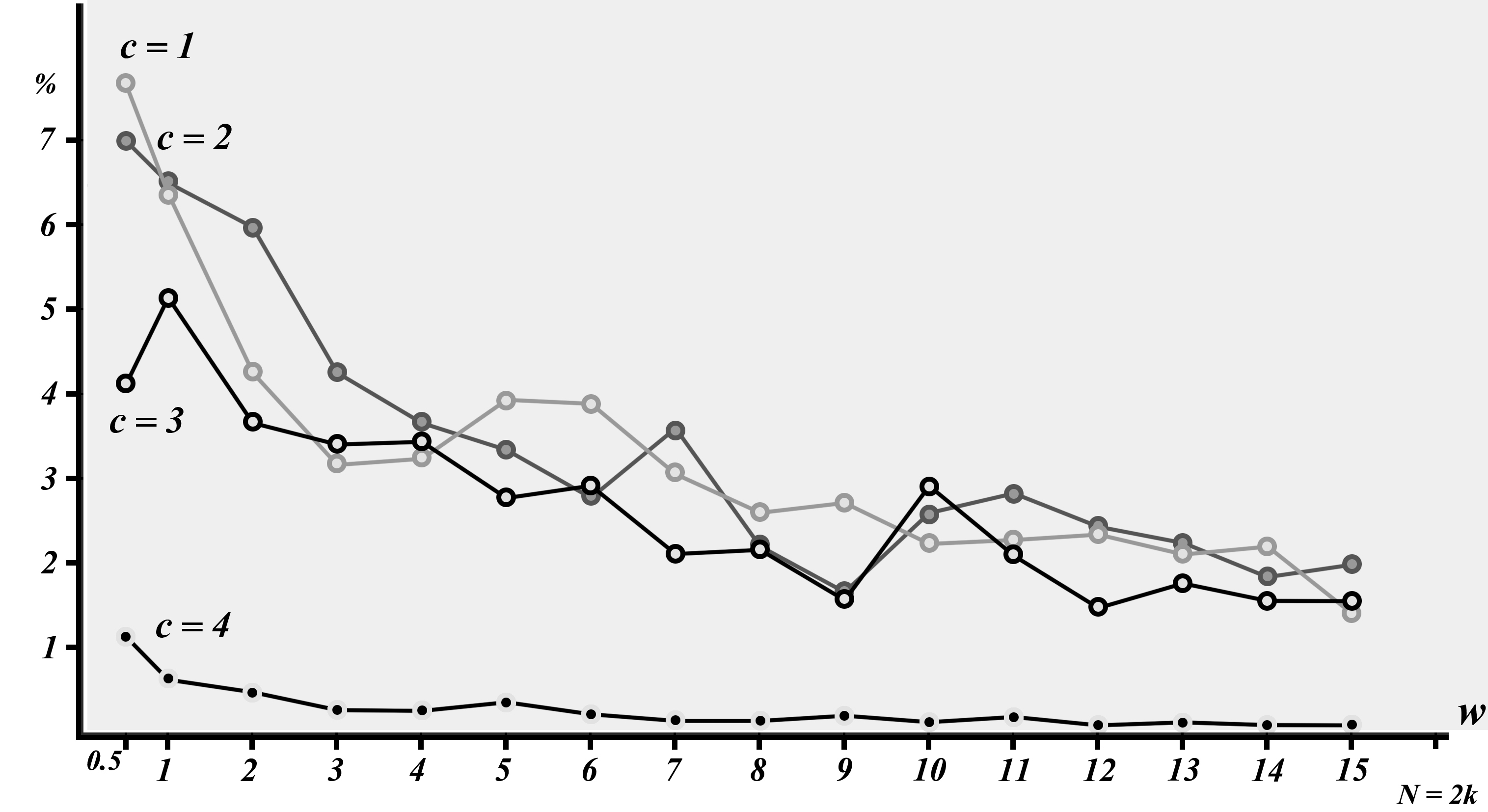}\\
Figure 9
\end{center}
~\newline
We can find that the case has extremely changed when the $c=4$. It shows much more off in plain and near to horizon than others, and there are many more results on top position. The case of enlarging gap alters our anticipation throughout that it made an approaching qualitative change. Particularly we let $c=5$, all results were on top position. This took us back to the static system and theorem 2 as well. Here we briefly call the case \emph{dynamical regularization}, it said that a dynamical system is presenting to a static system for calculation. 

For above case, a question has been emerged: is there any critical switch to let the instance turn to regularization upon hybrid method? We designed an exclusive trial to survey the case of dynamical regularization upon changing the density and cardinality of array $w$.

The testing object kept still. We let a number array $Y$ in front of the array $w$ and $Y=(y_{i})_{i=1}^{t}\colon y_{i}=i$ and it had a job to randomly assign its members to $w$, so as to take the gap in array $w$ to irregular. Afterword the program further randomly dispatched those numbers to weight table $W$ from array $w$. Therefore our trial turned out a relationship among two arguments, which they are the cardinalities about two arrays $Y$ and $w$ respectively. If each result be on top position for T times testing, we could have the ratio of $\eta =\vert Y\vert /\vert w\vert$, so that we plotted the broken lines about this ratio.

\begin{center}
\includegraphics[height=45mm]{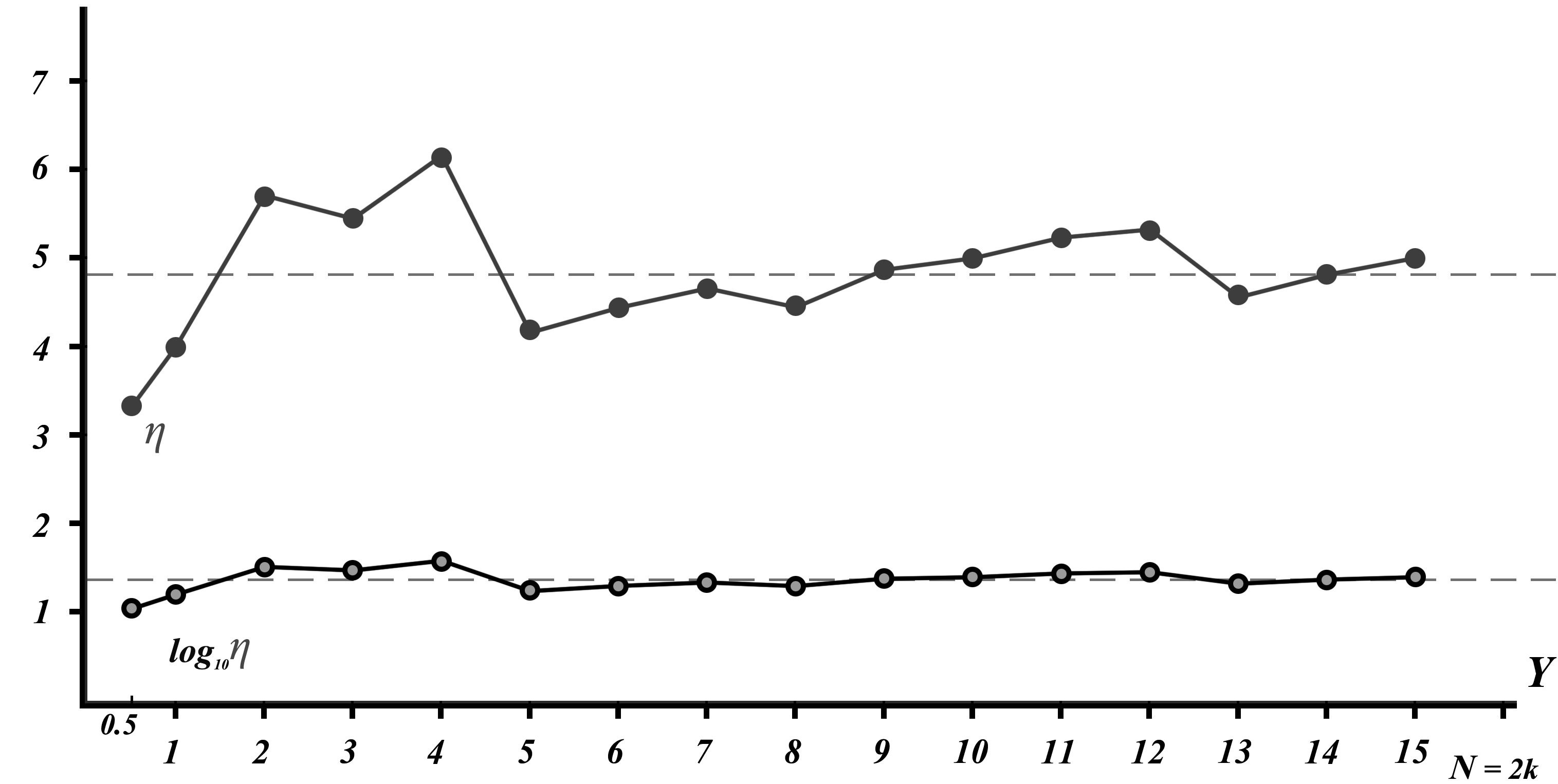}\\
Figure 10
\end{center}
~\newline
On the $x$-axis, the unit number is the cardinality about $Y$. the upper line is $\eta$ and the below one is $\log_{10}\eta$, that dealing can make the line to more plain. Two dash lines are their average horizon. Indeed, the diagram seems to demonstrate a critical argument for the instance, for argument $\eta$, we might use the form $(w_{max}- w_{min})/\vert w\vert$ instead of the above form to express the approaching meaning for cast it. That lets us might use a simple method to evaluate the accuracy made by the hybrid approach, moreover we can likewise utilize it to investigate the complexity of system. For example, we can survey a concrete instance and obtain a relationship about this argument and the data distribution in the system, so that we might gain an overall score in real time for govern system.\\
~\newline
\textbf{Discussion. }For above contents, we actually reported a group of trials to exhibit how to calculate a dynamical system. But for much more concrete instance, indeed we must have to announce that we merely exhibited an essential framework. For example, the whole weights in above dynamical system have been known, maybe in practice they are of uncertain. In this section, some arguments still are same as what were described in the static system, which contains the scale of data, the density and figure. Eventually, our research line was similar to partial derivative for some variables being frozen. The notion must be emphasized that the verification method we developed in the general sense takes our research work finish in the range of dynamical system. In particular, we even could not find out the case of dynamical regularization without the avail of EBOGP method, any way it is hard to imagine before, that in a chaos system we either obtained an outcome likes in static system.  

Turn back to the core theme, how to value the accuracy of outcomes made by hybrid method in a dynamical system. Frankly it said there is not any theories to support our trials and assumptions, thus our solution is just said that we stand at the highland of shortest path and work under a wing, where the case of polygon inequality relationship may be popular in system. On the other hand, it seems to be no worth for sought the absolutely optimal path in the dynamical system. Why? For the above example, any index does not mean the many of a weight on an arc which case we said uncertain. At a extreme notion, although we seek all paths out, the outcome is also possible not in exact because the uncertain factor maybe request the exploring job need include all trails which may contains circuits made of repeated paths. Hence the job becomes insignificance however for computing resources or task own. 

We must say that in last trial, it seems to divulge the dynamical system actually is a collection of static system with a temporal dimension. A question has been raised that is there a characteristics value in such dynamical system? If it exists, then we how to describe these relationships among density, the scale of data and the figure? Thus we still need to survey many more distinct instances for more evidences.

\section{Optimizing Accuracy}
We have issued the rule above that the accuracy is relative to anti event in static system. But it is obstacle that how the program detects whether the anti event happens and in where, however we cannot perform the BOTS to demonstrate all possible. By the theorem 2, we can find that the potential optimal path must be longer than the result made by hybrid method. Namely, there always possibly is a more suitable leaf in the native or later region. Instead, the first exploring amid those fore-leaves might be just viewed as an initial step. In addition to source, we might promote each endpoint to evolve by its own through screening all ownership leaves and their weights, not qualify those in which regions. 

The program can iteratively circularly update the arrays T and P, till not any total weight no longer be sought out to change again, said the evolution stops. In fact, the new method merely enlarges the search range than hybrid approach, that on encode level the rest portion is not changed. We use a form to present the logic on this step.
\begin{center}
$w_u=\text{min}\bigcup_{\in\beta}\left(w_{x} + w(x,u)\right)\text{ for }x\in L_{\in\beta}(u)\text{ and }w(x,u)\in W_{\in\tau}.$
\end{center}
~\newline
This method has got a useful idea to overcome the anti event in application any way. In our real world, sometimes the shortest path maybe implies worse in some cases. For example there is a block\textquoteright s traffic in jam and you are unfortunately to be ready to pass through that block as your previous plan. At this time, you must want to escape this under-nose nut through invoking overall traffic resources. Namely, your car needs to detour to other nice roads. This is a vivid example for anti event. 

On strategy level, we solve this problem upon three stages: firstly use partition to reveal the geometry layout; second seek an optimal path among those shortest paths; third overcome the possible anti event by each node in a sustainable evolution.

If we stand at the accuracy level, we could find the new method likely seems to achieve an extreme exact in static system and perhaps in dynamical system. But we either lack more proper evidence to stand this proposal. Author had done some trials to compare with hybrid approach. The prime problem is that the degree of overall improvement is not more than 5\%. Else for this method there is a drawback about runtime complexity, so as to we say it is not matured for this reason. Although the new method is actually to iteratively invoke the similar exploring model to promote those endpoints in evolution, but we cannot value the iterative times. Through trials we only said the quantity involving to $n$, since that the algorithm complexity may be $O(mn^2)$ only if the operation of querying weight table is a constant. On a popular sense, the method can take each endpoint in a sustainable evolution\cite{9} by a group criterion, we call it \emph{evolutionary hybrid Dijkstra\textquoteright s algorithm}.

\section{Summary}
In this paper, we introduced a universal solution to solve an ancient problem. All key points should be focused on the geometry layout. Well, we likewise used this idea to seek out the importance concept in our theoretical system, which is the polygon inequality. As well we even further found it was critical factor to influence the process of decision in greedy idea. This is a big harvest for us, but we have to recognize the fact that though this method has been at the frontline, but there still lack even more theories and trials to stand it and takes it completely matured. 

At the encode level, we eventually govern the decision stream pass through those nodes and optimize that algorithm, so that we might get rid of the domination from \emph{black-box} program, moreover we utilized a universal method to compute for dynamical system, and study the system to obtain various conclusions from trials.

What we need to emphasize is the context about this solution, it must be upon the least nodes and then to obtain an extreme result. Yes, you can reform algorithm to obtain other class outcome by altering some criteria. Note that for evolutionary hybrid Dijkstra\textquoteright s algorithm, you must be careful to select your strategy, because to a monotonic function, an extreme criterion perhaps means the outcome towards to infinity, which possible takes program into an endless loop and data overfill memory.

\section{Future Works}
At least author deem we must clarify the complexity of evolutionary algorithm to take it mature. Then after we can do even more trials to reveal more exactly features of complex system, special to dynamical system. Our goal is through more evidences to profile some essential principles to support our application and theory. Meanwhile we can define various concrete instances and let our research job in a correct way.\\
\newline
\textbf{Application.} It is natural thing to pour this idea into robot’s OS to deal the smart commute in a complex environment. For example, when a self-driving vehicle commutes in a city, he might record the data of traffic state for each road. This case likes our trial for dynamical system, while he has task in commute, he can plot an optimal lane depending on his memory about traffic, and share the information among his crew at once. 

Another case is the robots roaming in a narrow space. They must avoid such state happen, which they jam several channels, but for other channels either are idle. That said when a serve plans a crew touring in this narrow space, he has to view somewhere as resources for commute job. He can demand each crew member need to apply to him for using some resources, and estimate the lane for each applicant with the reserve table which other members have been checked in previously. 

Even to the array P and array T, they depict a layout likes a tree under some criterions. The only entrance is the source, so we can view this layout is a service network according to certain policy. For example, this topological layout might describe a plenty of servers distribution for cloud computing, a cryptosystem or an abstracted coding or compiling system. Summarily it represents an optimized layout abstracted from original defined figure. And the layout can be changed by difference entrance and relationship definition for nodes.


\begin{thebibliography}{99}
\bibitem{1}Tan, Yong. "Construct Graph Logic" CoRR abs/1312.2209 (2013) .
\bibitem{2}Tan, Yong. "Analyzing Traffic Problem Model With Graph Theory Algorithms" CoRR abs/1406.4828 (2014).
\bibitem{3}Dijkstra, E. W. (1959). "A note on two problems in connexion with graphs" (PDF). Numerische Mathematik 1: 269–271. doi:10.1007/BF01386390.
\bibitem{4}Bang-Jensen, Jørgen; Gutin, Gregory (2000). "Section 2.3.4: The Bellman-Ford-Moore algorithm". Digraphs: Theory, Algorithms and Applications (First ed.).
\bibitem{5}Cormen, Thomas H.; Leiserson, Charles E.; Rivest, Ronald L. (1990). Introduction to Algorithms (1st ed.). MIT Press and McGraw-Hill. See in particular Section 26.2, "The Floyd–Warshall algorithm", pp. 558–565 and Section 26.4, "A general framework for solving path problems in directed graphs", pp. 570–576.
\bibitem{6}Johnson, Donald B. (1977), "Efficient algorithms for shortest paths in sparse networks", Journal of the ACM 24 (1): 1–13, doi:10.1145/321992.321993.
\bibitem{7}Mohamed A. Khamsi, William A. Kirk (2001). "1.4 The triangle inequality in Rn". An introduction to metric spaces and fixed point theory. 
\bibitem{8}Holmes, Philip. "Poincaré, celestial mechanics, dynamical-systems theory and “chaos”." Physics Reports 193.3 (1990): 137-163.
\bibitem{9}Ashlock, D. (2006), Evolutionary Computation for Modeling and Optimization, Springer.
\end{thebibliography}
\end{document}